\documentclass{llncs}

\usepackage{enumerate}
\usepackage{amssymb,amsmath}
\usepackage{enumitem}
\usepackage{verbatim}
\usepackage{hyphenat}
\usepackage{color}
\usepackage{lineno}

\usepackage{url}

\usepackage{graphicx}
\usepackage{float}
\usepackage{multicol}
\usepackage{mathtools}
\usepackage{textcomp}

\newcommand{\cent}[0]{\mbox{\textcent}}
\newcommand{\dollar}[0]{\$}

\newtheorem{fact}{Fact}

\newcommand{\usqr}{\mathtt{USQUARE}}
\newcommand{\upower}{\mathtt{UPOWER64}}
\newcommand{\dimatwo}{\mathtt{DIMA2}}

\newcommand{\coinI}{\mathfrak{coin}_I}
\newcommand{\coinL}{\mathfrak{coin}_L}

\newcommand{\sigmastar}{\Sigma^{*}}
\newcommand{\tildesigma}{\tilde{\Sigma}}
\newcommand{\tildegamma}{\tilde{\Gamma}}
\newcommand{\tildew}{\tilde{w}}

\newcommand{\directions}{ \{ \leftarrow,\downarrow,\rightarrow \} }

\begin{document}

\title{Probabilistic verification of all languages}

\author{Maksims Dimitrijevs \and Abuzer Yakary\i lmaz}
\institute{University of Latvia, Faculty of Computing \\  Rai\c na bulv\= aris 19, R\={\i}ga, LV-1586, Latvia
\\~\\
University of Latvia,
Center for Quantum Computer Science \\  Rai\c na bulv\= aris 19, R\={\i}ga, LV-1586, Latvia
\\ ~ \\
\textit{md09032@lu.lv, abuzer@lu.lv}}


\maketitle

\begin{abstract}
	We present three protocols for verifying  all languages: (i) For any unary (binary) language, there is a log-space (linear-space) interactive proof system (IPS); (ii) for any language, there is a constant-space weak-IPS (the non-members may not be rejected with high probability); and, (iii) for any language, there is a constant-space IPS with two provers where the verifier reads the input once. Additionally, we show that uncountably many binary (unary) languages can be verified in constant space and in linear (quadratic) expected time. 
\end{abstract}

\section{Introduction} 
	When using arbitary real-valued transitions, quantum and probabilistic machines can recognize uncountably many languages with bounded error \cite{ADH97,DY16A,DY17,SayY17} because it is possible to encode an infinite sequence as a transition value and then determine its $ i $-th bit by using some probabilistic or quantum experiments. 

	For a given alphabet $\Sigma$, all strings can be lexicographically ordered. Then, for a language $ L  $ defined on $\Sigma$, the membership information of all strings with respect to $ L $ can be stored as a single probability, say $ p_L $, such that the membership value of $ i $-th string can be stored at the $ O(i) $-th digit of $ p_L $. Thus, we can determine whether a given string is in $ L $ or not (with high probability) by learning its corresponding digit in $p_L$. For this purpose, we can toss a coin coming head with probability $ p_L $ exponential many times (in $ i $) and count the total number of heads to guess the corresponding digit with high probability. By using this method, we can easily show that any unary (binary) language can be recognized in linear (exponential) space. 
	
	In this paper, we present better space bounds by using interactive proof systems (IPSs). By using fingerprinting method \cite{Fre83}, we exponentially reduce the above space bounds by help of a single prover. Then, by adopting the one-way\footnote{The verifier always sends the same symbol and so the only useful communication can be done by the prover.} protocol given in \cite{CL89} for every recursive enumerable language, we present a constant-space weak-IPS  with two-way communication for every language. Here ``weak'' means that the non-members may not be rejected with high probability. In \cite{FS89}, it was shown that one-way probabilistic finite automata can simulate the computation of a Turing machine by communicating with two provers. We also modify this protocol and extend the same result for every language. 
	
	Currently we left open whether there is a constant-space IPS for every language or not. (Remark that the answer is positive for Arthur-Merlin games with constant-space quantum verifiers \cite{SayY17}.) But, we obtain some other strong results for constant-space IPSs such that the binary (unary) languages having constant-space linear-time (quadratic-time) IPSs form an uncountable set. Remark that it is also open whether constant-space probabilistic machines can recognize uncountably many languages or not \cite{DY16A}.

In the next section, we present the notations and definitions to follow the rest of the paper. Then, we present our results in two sections. The Section \ref{sec:verification-all} is dedicated to the verification of all languages, while in the Section \ref{sec:verification-uncountable} we present our results for uncountably many languages. The latter section is divided into two subsections. In Section \ref{sec:languages}, we present our constant-space protocols for two unary and one binary languages that are used in Section \ref{sec:verifications}. In Section \ref{sec:verifications}, we present our constant-space protocols verifying uncountably many languages.

\section{Background}
\label{sec:def}

We assume the reader is familiar with the basics of complexity theory and automata theory. We denote the left and the right end-markers as $ \cent $ and $ \dollar $, and the blank symbol as $\#$. The input alphabet not containing symbols $\cent$ and $\dollar$ is denoted by $ \Sigma $, $ \tildesigma $ denotes the set $ \Sigma \cup \{ \cent,\dollar \} $, the work tape alphabet not containing symbol $\#$ is denoted by $ \Gamma $, $ \tildegamma $ denotes the set $ \Gamma \cup \{ \# \} $, and the communication alphabet is denoted by $\Upsilon $. $ \sigmastar $ denotes set of all strings (including the empty string ($\varepsilon$)) defined over $\Sigma$. We order the elements of $ \sigmastar $ lexicographically and then represent the  $i$-th element by $ \sigmastar(i) $, where $ \sigmastar(1) = \varepsilon$. For any natural number $i$, $bin(i)$ denotes the unique binary representation and $(bin(i))^r$ denotes the reverse binary representation. For any given string $w \in \sigmastar$, $ |w| $ is its length, $ w[i] $ is its $ i $-th symbol from the left ($ 1 \leq i \leq |w| $), $ \tildew = \cent w \dollar $, and $lex(w)$ denotes the lexicographical number of $w$, such that $lex(\sigmastar(i)) = i$ for any $i>0$. Any given input, say $w$, is always placed on the input tape as $ \tildew $.

An interactive proof system (IPS) \cite{GMR89,Bab85} is composed by a prover ($P$) and a (probabilistic) verifier ($V$), denoted as pair $(P,V)$, who can communicate with each other. The aim of the verifier is to make a decision on a given input and the aim of the prover (assumed to have unlimited computational power) is to convince the verifier to make positive decision. Thus, the verifier should be able to verify the correctness of the information (proof) provided by the prover since the prover may be cheating when the decision should be negative.

In this paper, we focus on memory-bounded verifiers \cite{DS92} and our verifiers are space--bounded probabilistic Turing machines. The verifier has two tapes: the read-only input tape and the read/write work tape. The communication between the prover and verifier is done via a communication cell holding a single symbol. The prover can see only the given input and the symbols written on the communication cell by the verifier. The prover may know the program of the verifier but may not know which probabilistic choices are done by the verifier. Since the outcomes of probabilistic choices are hidden from the prover, such IPS is called private-coin. If the probabilistic outcomes are not hidden (sent via the communication channel), then it is called public-coin, i.e. the prover can have complete information about the verifier during the computation. Public-coin IPS is also known as Arthur-Merlin games \cite{Bab85}.

A space-bounded probabilistic Turing machine (PTM) verifier is a space-bounded PTM extended with a communication cell. Formally, a PTM verifier $ V $ is a 8-tuple
\[
	(S,\Sigma,\Gamma,\Upsilon,\delta,s_1,s_a,s_r),
\]  
where, S is the set of states, consisting of three disjoint sets: $S_r$ is the set of reading states, $S_c$ is the set of communicating states, and $S_h$ is the set of halting states; $s_1 \in S_r$ is the initial state, $ s_a \in S_h $ and $ s_r \in S _h$ ($s_a \neq s_r$) are the accepting and rejecting states, respectively, and $ \delta $ is the transition function composed by $ \delta_c $ and $\delta_r$ that are responsible for the transitions when in a communicating state and in a reading state, respectively. There is no transition from $s_a$ or $s_r$ since when V enters a halting state ($s_a $ or $s_r$), the computation is terminated.

When in a communicating state, say $ s \in S_c $, $V$ writes symbol $ \tau_s \in \Upsilon $ on the communication cell and the prover writes back a symbol, say $ \tau \in \Upsilon $. Then, $V$ switches to state $ s' = \delta_c(s,\tau) \in S $.

When in a reading state, $V$ behaves as an ordinary PTM:
\[
\delta_r: S_r \times \tildesigma \times \tildegamma \times S \times \tildegamma \times \directions \times \directions \rightarrow [0,1].
\] 
That is, when $V$ is in reading state $ s \in S_r $, reads symbol $ \sigma \in \tildesigma $ on the input tape, and reads symbol $ \gamma \in \tildegamma $ on the work tape, it enters state $ s' \in S $, writes $ \gamma' \in \tildegamma $ on the cell under the work tape head, and then the input tape head is updated with respect to $ d \in \directions $ and the work tape head is updated with respect to $ d' \in \directions $ with probability
\[
	\delta(s,\sigma,\gamma,s',\gamma',d,d'),
\]
where ``$ \leftarrow $'' (``$\downarrow$'' and ``$\rightarrow$'') means the head is moved one cell to the left (the head does not move and the head is moved one cell to the right). To be a well-formed PTM, the following condition must be satisfied: for each triple $(s,\sigma,\gamma) \in S_r \times \tildesigma \times \tildegamma $,
\[
	\sum_{s' \in S,\gamma' \in \tildegamma,d \in \directions,d' \in \directions} \delta(s,\sigma,\gamma,s',\gamma',d,d') = 1.
\]
In other words, all outgoing transitions for the triple $ (s,\sigma,\gamma) $ have total probability of 1.

The space used by $V$ on $w$ is the number of all cells visited on the work tape during the computation with some non-zero probability. The verifier $V$ is called to be $ O(s(n)) $ space bounded machine if it always uses $ O(s(n)) $ space on any input with length $n \geq 0$.
The language $ L \subseteq \Sigma^* $ is verifiable by $V$ with error bound $\epsilon <\frac{1}{2} $ if
\begin{enumerate}
	\item there exists an honest prover $P$ such that any $ x \in L $ is accepted by $V$ with probability at least $ 1-\epsilon $ by communicating with $P$, and,
	 \item any $ x \notin L $ is always rejected by $V$ with probability at least $ 1-\epsilon $ when communicating with any possible prover ($P$*).
\end{enumerate}

The first property is known as completeness and the second one known as soundness. Generally speaking, completeness means there is a proof for a true statement and soundness means none of the proof works for a false statement. The case when every member is accepted with probability 1 is also called as perfect completeness.

It is also said that there is an IPS ($P$, $V$) with error bound $ \epsilon $ for language $ L $. Remark that all the time and memory bounds are defined for the verifier since we are interested in the verification power of the machines with limited resources. 

We also consider so-called weak IPS, which is obtained by replacing the condition 2 above by the following condition.

\begin{itemize}
    \item [$2'$.] Any $ x \notin L $ is accepted by $V$ with probability at most $ \epsilon $ when communicating with any possible prover ($P$*).
\end{itemize}

Therefore, in weak IPSs, the computation may not halt with high probability on non-members, or, in other words, each non-member may not be rejected with high probability.

Due to communications with the prover, the program of the verifier with the possible communications is called protocol. A private-coin protocol is called one-way if the verifier always sends the same symbol to the prover. In this case, we can assume that the prover provides a single string (possible infinite) and this string is consumed in every probabilistic branch. It is also possible that this string (certificate) is placed on a separate one-way read-only tape (certificate tape) at the beginning of the computation and the verifier can read the certificate from this tape.

We also consider interactive proof systems with two provers \cite{BGKW88,FST88}. IPS with two provers is composed by two provers ($P_1$, $P_2$) and a probabilistic verifier ($V$), denoted as ($P_1$, $P_2$, $V$). The verifier has a communication channel with each prover and one prover does not see the communication with the other prover. The verifier $V$ in such IPS has different transition function $\delta_c$. When in a communicating state, say $ s \in S_c $, $V$ writes symbol $ \tau_{s1} \in \Upsilon $ on the communication cell of the first prover ($P_1$) and $ \tau_{s2} \in \Upsilon $ on the communication cell of the second prover ($P_2$), and the provers $P_1$ and $P_2$ write back symbols, say $ \tau_1 \in \Upsilon $ and $ \tau_2 \in \Upsilon $, respectively. Then, $V$ switches to state $ s' = \delta_c(s,\tau_1,\tau_2) \in S $.

There are different models of IPS with two provers. In Multi-Prover model by \cite{BGKW88} both provers collaborate such that both of them are honest, or both of them are cheating. In Noisy-Oracle model by \cite{FST88} both provers oppose each other such that at least one of them is honest, and other may be cheating. The latter model can also be formalized as a debate system (see e.g. \cite{DSY15}) where the second prover is called as refuter. In this paper we consider the IPSs with two provers working for both models equally well.

Without communication, a verifier is a PTM. For a PTM, we simply remove the components related to communication in the formal definition (and so PTM does not implement any communicating transition). When there is no communication, then we use term recognition instead of verification.

A PTM without work tape is a two--way probabilistic finite state automaton (2PFA). Any verifier without work tape is a constant-space verifier or 2PFA verifier.

A 2PFA can be extended with $ k $ integer counters - such model is called two-way probabilistic automaton with $k$ counters (2P$k$CA). At each step of computation 2P$k$CA can check whether the value of each counter is zero or not, and then,  as a part of a transition, it can update the value of each counter by a value from $ \{-1,0,1\} $.

A two-way model is called sweeping if the direction of the input head can be changed only on the end-markers. If the input head is not allowed to move to the left, then the model is called ``one-way''.

We denote the set of integers $ \mathbb{Z} $ and the set of positive integers $ \mathbb{Z}^+ $. The set $ \mathcal{I} = \{ I \mid I \subseteq \mathbb{Z^+} \} $ is the set of all subsets of positive integers and so it is an uncountable set (the cardinality is $ \aleph_1 $) like the set of real numbers ($ \mathbb{R} $). The cardinality of $ \mathbb{Z} $ or $ \mathbb{Z^+} $ is $ \aleph_0 $ (countably many). 

The membership of each positive integer for any $ I \in \mathcal{I} $ can be represented as a binary probability value:
\[
p_I = 0.x_1 0 1 x_2 0 1 x_3 0 1 \cdots x_i 0 1 \cdots,~~~~ x_i = 1 \leftrightarrow i \in I.
\]  
Similarly, the membership of each string for language $ L \subseteq \sigmastar $ is represented as a binary probability value:
\[
	p_L = 0.x_1 0 1 x_2 0 1 x_3 0 1 \cdots x_i 0 1 \cdots,~~~~ x_i = 1 \leftrightarrow \Sigma^*(i) \in L.
\]
The coin landing on head with probability $ p_I $ (resp., $p_L$) is named as $\coinI$ (resp., $\coinL$).

\section{Verification of all languages}
\label{sec:verification-all}

We start with a basic fact presented in our previous paper \cite{DY16A}.
\begin{fact}
	\label{fact:DY16A}
	\cite{DY16A} 
	Let $x=x_1 x_2 x_3 \cdots$ be an infinite binary sequence. If a biased coin lands on head with probability  $p = 0. x_1 0 1 x_2 0 1 x_3 0 1 \cdots$, then the value $x_k$ is determined correctly with probability at least $\frac{3}{4}$ after $64^k$ coin tosses, where $ x_k $ is guessed as the $ (3k+3) $-th digit of the binary number representing the total number of heads after the whole coin tosses. 
\end{fact}
Then, the following results can be obtained straightforwardly.
\begin{theorem}
	\label{thm:all-unary-recognize}
	Any unary language is recognized by a linear-space PTMs with bounded error.
\end{theorem}
\begin{proof}
	Let $ \Sigma = \{a\} $ be our alphabet. For any unary language $L \subseteq \Sigma^*$, we can design a PTM for $ L $, say $ P_L $, that uses $ \coinL $.
	
Let $ w = a^n $ be the given input for $ n \geq 0 $.  The PTM $ P_L $ implements the procedure described in Fact \ref{fact:DY16A} in a straightforward way and gives its decision accordingly, which will be correct with probability not less than $ \frac{3}{4}$. The machine only uses linear-size binary counters to implement $ 64^k $ coin tosses and to count the number of heads (for unary $ L $ and $ \coinL $, $k = n+1$). By repeating the same procedure many times (not depending on $n$), the success probability is increased arbitrarily close to 1.
\qed 
\end{proof}

\begin{remark}
	\label{rem:k-ary}
	Let $ L \subseteq \Sigma^* $ be a $k$-ary language for $ k>1 $, where $ \Sigma = \{a_1,\ldots,a_k\} $. For any given $ k $-ary string $ w \in \Sigma^* $, let $ x_l $ represent its membership bit in $ p_L $. Then, by using the exactly the same algorithm given in the above proof, we can determine $ x_l $ correctly with high probability. However, $ l $ is exponential in $ |w| $ and so the PTM uses exponential space.
\end{remark}	
	 
\begin{corollary}
	\label{cor:all-k-ary-exp}
	Any $k$-ary ($ k>1 $) language is recognized by an exponential-space PTMs with bounded error.
\end{corollary}

When interacting with a prover, we can obtain exponentially better space bounds. For this purpose, we use probabilistic fingerprint method: For comparing two $ t $-bit numbers, say $ n_1 $ and $n_2$, we can randomly pick a $ (c\log t) $-bit  prime number $ p $ for some positive integer $ c $ and compare $ n_1'=n_1 \mod p $ with $ n_2' = n_2 \mod p $. It is known that \cite{Fre83} if $ n_1=n_2 $, then clearly $ n_1' = n_2' $, and, if $ n_1 \neq n_2 $, then  $ n_1' \neq n_2' $ with high probability depending on number $ c $ as specified below.

\begin{fact}	
	\label{fact:Fre83}
	\cite{Fre83}
	Let $P_1 (m)$ be the number of prime numbers not greater than $2^{ \lceil log_2 m \rceil }$, $P_2 (l,N',N'')$ be the number of prime numbers not greater than $2^{ \lceil log_2 l \rceil }$ and dividing $|N'-N''|$, and $P_3 (l,m)$ be the maximum of $P_2 (l,N',N'')$ over all $N'<2^m$, $N'' \leq 2^m$, $N' \neq N''$.
		Then, for any $\epsilon > 0$, there is a natural number $c$ such that $\lim_{m\to\infty} \frac{P_3 (cm,m)}{P_1 (cm)} < \epsilon$.  
\end{fact}

\begin{theorem}
	\label{thm:all-unary}
	Any unary language $ L \subseteq \{a\}^* $ is verified by a log-space PTM with bounded error.
\end{theorem}
\begin{proof}
	The protocol is two-way. Let $ w = a^n $ be the given input for $ n>0 $. (The decision on the empty string is given deterministically.) Remember that the membership bit of $ a^n $ is $ x_{n+1} $ in $ p_L $. Let $k = n+1 \geq 2$. We pick a value of $ c $ (see Fact \ref{fact:Fre83}) satisfying the error bound $ \frac{1}{8} $.

	The protocol has three phases. In the first phase, there is no communication. The verifier picks two random $ (c \cdot (4 \cdot \log k)) $-bit prime numbers, say $ p_1 $ and $ p_2 $, and then, it calculates and stores $ r_1 = 64^k \mod p_1 $ in binary on the work tape. The verifier also prepares two binary counters $ C_1 $ and $ C_2 $ for storing the total number of coin tosses in modulo $ p_1 $ and the total number of heads in modulo $ p_2 $, respectively, and one ``halting'' counter $ C_h = 0 $.
	
	In the second phase, the verifier asks from the prover to send $ a^{64^k}b $. Once the verifier receives symbol $ b $, the communication is ended in this phase. Therefore, we assume that the prover sends either the finite string $ y=a^mb $ for some $  m\geq 0 $ or an infinite sequence of $ a $'s. 
	
	For each $ a $ received from the prover, the verifier reads the whole input and adds one to $C_h$ with probability $ \left( \frac{1}{64} \right)^{k} $. We call it as halting walk. If $C_h = 8$, the verifier terminates the computation and rejects the input. If the computation is not terminated, the verifier tosses $ \coinL $, sends the result to the prover, and increases $C_1$ by 1. If the result is head, the verifier also increases $ C_2 $ by 1. When the second part is ended, the verifier checks whether previously calculated $ r_1 $ is equal to $ r'_1 = m \mod p_1 $, stored on $ C_1 $. If they are not equal, then the input is rejected. In other words, if the prover does not send $ 64^k $ $a$'s, then the verifier detects this with probability at least $ \frac{7}{8} $.
	
	Let $ r_2 $ be the binary value stored on $ C_2 $ and $ t $ be the total number of heads obtained in the second phase. In the third phase, the verifier asks from the prover to send $ (bin(t))^r $ (the least significant bits are first). By using the input head, the verifier easily reads the $ (3k+3) $-th bit of $ bin(t) $, say $ x'_{k} $, and also checks whether the length of $ bin(t) $ does not exceed $ (6k+1) $. Meanwhile, the verifier also calculates $ r'_2 = t \mod p_2 $. If the length of $bin(t) $ is greater than $ 6k+1 $ or $ r_2 \neq r'_2 $, then the input is rejected. In other words, if the prover does not send $ bin(t) $, then the verifier can catch it with probability $ \frac{7}{8} $. At the end of third phase, the verifier accepts the input if $ x'_{k} $ is 1, and, rejects it if $ x'_{k} $ is 0.
    
    With respect to Chebyshev's inequality, the value of $C_h$ reaches 8 after more than $16 \cdot 64^k$ $ a $'s with probability
    \[
    	Pr[|X-E[X]| \geq 9] \leq \frac{(\frac{1}{64^k}) \cdot (1-\frac{1}{64^k}) \cdot 16 \cdot 64^k}{9^2} < \frac{16}{81},
    \]
     where $E[X]$ is expected value of $C_h$. This bound is important, since, for $ m > 16 \cdot 64^k$, Fact \ref{fact:Fre83} cannot guarantee the error bound $\frac{1}{8}$. (Remember that prime numbers $ p_1 $ and $ p_2 $ do not exceed $2^{c \cdot (4 \cdot \log k)} \geq 2^{c \cdot (\log (6 \cdot k + 4))}$ for $k \geq 2$.)
	
	If $ w $ is not a member, then the accepting probability can be at most $ \frac{209}{648}$.
	\begin{itemize}
		\item  If $ m \neq 64^k  $, then the input is rejected with probability at least $ \frac{7}{8} - \frac{16}{81}$, and so the accepting probability cannot be greater than $ \frac{209}{648} $. 
		\item Assume that $ m=64^k $. If the prover does not send $ (bin(t))^r $, then the input is rejected with probability $ \frac{7}{8} $, and so the accepting probability cannot be greater than $ \frac{1}{8} $. 
		\item Assume that $ m=64^k $ and the verifier sends $ (bin(t))^r $, then the input is accepted with probability at most $ \frac{1}{4} $.
	\end{itemize}	
	The expected running time for the non-members is exponential in $ n $ due to the halting walks.
	
	If $ w $ is a member, then the honest prover sends all information correctly, and the verifier guesses $ x_{n+1} $ correctly with probability at least $ \frac{3}{4} $ if the computation is not terminated by halting walks. The probability of halting the computation (and rejecting the input) in the second phase is 
	\[
		Pr[|X-E[X]| \geq 7] \leq \frac{(\frac{1}{64^k}) \cdot (1-\frac{1}{64^k}) \cdot 64^k}{7^2} < \frac{1}{49}.
	\]
	Therefore, the verifier accepts $ w $ with probability at least $ \frac{143}{196} $. The expected running time for members is also exponential in $ n $.
	
	The verifier uses $ O(\log n) $ space and the success probability is increased by repeating the same algorithm.
\qed
\end{proof}

Due to Remark \ref{rem:k-ary}, we can follow the same result also for $ k $-ary languages with exponential increase in time and space.  

\begin{corollary}
	\label{cor:all-k-ary-linear}
	Any $k$-ary ($k>1$) language $ L \subseteq \{a_1,\ldots,a_k\}^* $ is verified by a linear-space PTM with bounded error.
\end{corollary}

Currently we do not know any better space bound for (strong) IPSs. On the other hand, by using more powerful proof systems, we can reduce space complexity to $ O(1) $. We first present a 1P4CA algorithm for any language.

\begin{theorem}
	\label{thm:four-counter-every-language}
	Any $k$-ary ($ k \geq 1 $) language $L$ is recognized by a 1P4CA $P_L$  with bounded error.
\end{theorem}
\begin{proof}
Let $ \Sigma = \{a_1,\ldots,a_k\} $ be the input alphabet with $k$ symbols, and for each $1 \leq i \leq k$, $lex(a_i) = i+1$. Let $w = w[1] w[2] \cdots w[n-1] w[n]$ ($|w| = n$) be the given input string. If $n=0$, then $ lex(w)=1 $. If $ n>0 $, then $ lex(w) $ is calculated as follows:
\[
lex(w) = 1 + \sum_{i=1}^{n} k^{i-1} + \sum_{i=1}^{n} (lex(w[i]) - 2) \cdot k^{n-i},
\]
since there are total $\sum_{i=1}^{n} k^{i-1}$ strings with length less than $n$, and there are total $\sum_{i=1}^{n} (lex(w[i]) - 2) \cdot k^{n-i}$ strings with length $n$ that are coming before $ w $ in the lexicographic order.

If $ w $ is the empty string, then the decision is given deterministically. In the following part, we assume that $ n>0 $.
Let $ C_j $ ($ 1 \leq j \leq 4 $) represent 
the value of $j$-th counter. At the beginning of computation $C_1 = C_2 = C_3 = C_4 = 0$. Firstly, $ P_L $ reads $ w $ and sets $ C_1 = lex(w) $ as follows. $P_L$ reads $w[1]$ and sets $C_1 = lex(w[1])$. After that, for each $ m \in \{2, \ldots,n\} $, $P_L$ reads $w[m]$ and multiplies the value $C_1$ by $k$ and increases it by $(2-k) + (lex(w[m])-2)$ with the help of other counters. 

  We claim that after reading $w$, $C_1 = lex(w)$. We prove this claim by induction on $m$. The basis, when $m=1$, is trivial, since $C_1 = lex(w[m])$.

Suppose that for some $m > 0$:
\[
C_1 = lex(w[1] w[2] \cdots w[m]) = 1 + \sum_{i=1}^{m} k^{i-1} + \sum_{i=1}^{m} (lex(w[i]) - 2) \cdot k^{m-i}
\]
Then $P_L$ reads $w_{[m+1]}$, and $C_1$ is updated to
\footnotesize
\begin{eqnarray*}
C_1 & = & (1 + \sum_{i=1}^{m} k^{i-1} + \sum_{i=1}^{m} (lex(w[i]) - 2) \cdot k^{m-i}) \cdot k + (2-k) + (lex(w[m+1])-2)
\\
	& = & k + \sum_{i=1}^{m} k^i + \sum_{i=1}^{m} (lex(w[i]) - 2) \cdot k^{m+1-i} + (2-k) + (lex(w[m+1])-2)
\\
	& = & 2 + \sum_{i=2}^{m+1} k^{i-1} + \sum_{i=1}^{m} (lex(w[i]) - 2) \cdot k^{m+1-i} + (lex(w[m+1])-2)
\\
	& = & 1 + \sum_{i=1}^{m+1} k^{i-1} + \sum_{i=1}^{m+1} (lex(w[i]) - 2) \cdot k^{m+1-i}.
\end{eqnarray*}
\normalsize
Thus, $ C_1 = lex(w[1] \cdots w[m] w[m+1]) $ and so the claim is proven.

After reading $ w $, $ P_L $ stays on the right end-marker and do the rest of the computation without moving the input head. Let   $ l = C_1 = lex(w) $. $P_L$ decreases $C_1$ by 1, and sets $C_2 = 64$ and $C_3 = 4 \cdot 8$. Then, until $ C_1 $ hits zero, $P_L$ decreases $ C_1 $ by 1, and then multiplies $C_2$ by 64 and $C_3$ by 8 with the help of 4th counter. When $C_1 = 0$, we have $C_2 = 64^l$, $C_3 = 4 \cdot 8^l$, and $C_4 = 0$.

After that $P_L$ tosses $\coinL$ $ 64^l $ times and guesses the value of $ x_l $ in $p_L$ by using the number of total heads. Remember that $ x_l = 1 $ if and only if $ w \in L $. Let $ H $ be the total number of heads in binary. Due to Fact \ref{fact:DY16A}, the ($3l+3$)-th bit of $ H $ is $ x_l $ with probability at least $ \frac{3}{4} $. Remark that when counting the total number of heads, the ($3l+3$)-th bit of $ H $ is flipped after each $ T = 2^{3l+2} = 4 \cdot 8^l $ number of heads. Thus, by switching the values of the 3rd and the 4th counters, $ P_L $ determines each block of $ T $ heads. By using its internal states, $ P_L $ keeps a candidate value for $ x_l $, say $ x'_l $. $P_L$ sets $ x'_l = 0 $ before the coin tosses and updates it to $ 1-x_l $ after each $ T $ number of heads. 

At the end of the coin tosses, if $ x'_l = 0 $ (resp., $ x'_l = 1 $), the input is rejected (resp., accepted). The decision will be correct with probability at least $ \frac{3}{4} $.
\qed
\end{proof}

It is well-known fact that any recursive language is recognized by a deterministic automata with two counters \cite{Min61}. Based on this fact, Condon and Lipton \cite{CL89} proved that for any recursively enumerable language $L$, there is a one-way weak-IPS with a 2PFA verifier by presenting a protocol that simulates the computation of a 2P2CA on a given input. We extend this protocol for the verification of any language by also using two-way communication.

\begin{theorem}
	\label{thm:weak-verify}
	There is a weak-IPS $ (P,V) $ for any language $L$ where $ V $ is a sweeping PFA.
\end{theorem}
\begin{proof}
For any language $L$, $V$ simulates the algorithm of 1P4CA $P_L$ described in the proof of Theorem \ref{thm:four-counter-every-language} on the given input $w$ and asks the prover to store the contents of four counters. 

  For each step of $ P_L $, $V$ interacts with the prover. First, $V$ asks from the prover the values of four counters as $a^{s_1} b^{s_2} c^{s_3} d^{s_4} e$, where $s_j$ is the value of $j$-th counter ($ 1 \leq j \leq 4 $). Second, $ V $ implements the transition of $ P_L $ based on the current state, the symbol under the input head, the probabilistic outcome, and the status of the counters. Then, $V$ updates the state and head position by itself and sends the updates on the values of the counters to the prover as  $f_1 f_2 f_3 f_4$. Here $ f_j \in \{-1,0,1\} $ and it means that the verifier asks from the prover to add $f_j $ to the value of the $j$-th counter, where $ 1 \leq j \leq 4 $.

The communication between the verifier and the prover is two-way since $ P_L $ may toss $ \coinL $ during its computation. Without loss of generality, we pick an arbitrary computation path of $ P_L $. For this path,  let 
\[
	c_V = f_{1,1} f_{1,2} f_{1,3} f_{1,4} ~ f_{2,1} f_{2,2} f_{2,3} f_{2,4} ~ \cdots ~ f_{i,1} f_{i,2} f_{i,3} f_{i,4} ~ \cdots  
\] be the string representing the messages sent by the verifier and let 
\[
	c_P = a^{s_{1,1}} b^{s_{1,2}} c^{s_{1,3}} d^{s_{1,4}} ~ e ~ a^{s_{2,1}} b^{s_{2,2}} c^{s_{2,3}} d^{s_{2,4}} ~ e ~ \cdots ~ e ~ a^{s_{i,1}} b^{s_{i,2}} c^{s_{i,3}} d^{s_{i,4}} ~ e ~ \cdots
\] 
be the string representing the messages sent by the prover. The verifier $ V $ can check the validity of $ c_P $ as described below. For each $ i $, $ V $ can compare $ s_{i,1} $, $ s_{i,2} $, $ s_{i,3} $, and $ s_{i,4} $ with $ s_{i+1,1} $, $ s_{i+1,2} $, $ s_{i+1,3} $, and $  s_{i+1,4}$, respectively, by using the values $ f_{i,1} f_{i,2} f_{i,3} f_{i,4} $. 

Let $y < \frac{1}{2}$ be the parameter to determine the error bound in the following checks. For the validity check of $c_P$, $V$ makes four comparisons in parallel such that one comparison is responsible for one counter. During the $j$-th comparison ($1 \leq j \leq 4$), $V$ creates two paths with equal probabilities. In the first path, $V$ says ``$A$'' with probability, denoted by $ Pr[A_j] $, 
\[
	Pr[A_j] = \prod_{i=1}^{g-1} y^{2s_{i,j}+2(s_{i+1,j} - f_{i,j} )} ,
\]
 in the second path, it says ``$R$'' with probability, denoted by $ Pr[R_j] $, 
\[
	Pr[R_j] = \prod_{i=1}^{g-1} \dfrac{y^{4s_{i,j}}+y^{4(s_{i+1,j} - f_{i,j})}}{2},
\]
where $g$ is the total number of computational steps of $ P_L $. Here each comparison executes two parallel procedures such that the first procedure produces the probabilities for odd $i$'s and the second procedure produces the probabilities for even $i$'s.

Once the simulation of $ P_L $ is finished, $ V $ says only ``A''s in all comparisons with probability 
\[
	Pr[A] = Pr[A_1] \cdot Pr[A_2] \cdot Pr[A_3] \cdot Pr[A_4] 
\]
and 
$ V $ says only ``R''s in all comparisons with probability 
\[
	Pr[R] = Pr[R_1] \cdot Pr[R_2] \cdot Pr[R_3] \cdot Pr[R_4]. 
\]

It is easy to see that if $ s_{i,j} = ( s_{i+1,j} - f_{i,j} )$ for each $i$ and $j$, then $ Pr[A] = Pr[R] = \prod_{i=1}^{g-1} y^{4(s_{i,1} + s_{i,2} + s_{i,3} + s_{i,4})} $. On the other hand, if $ s_{i,j} \neq ( s_{i+1,j} - f_{i,j} ) $ for some $i$ and $j$, then 
\[
	\frac{Pr[R]}{Pr[A]} \geq \frac{y^{2s_{i,j}-2(s_{i+1,j} - {f_{i,j}})}}{2} + \frac{y^{2(s_{i+1,j} - {f_{i,j}})-2s_{i,j}}}{2} > \frac{1}{2y^2}
\]
since either $ (2s_{i,j}-2(s_{i+1,j} - {f_{i,j}})) $ or $ (2(s_{i+1,j} - {f_{i,j}})-2s_{i,j}) $ is a negative even integer.

If $c_P$ is finite (a cheating prover may provide an infinite $ c_P $), then $ V $ gives a positive decision for $c_P$ with probability $ Pr[A] $ and negative decision for $c_P$ with probability $ (y \cdot Pr[R]) $. Hence, if $ c_P $ is valid, then the probability of positive decision is $ \frac{1}{y} $ times of the probability of negative decision. 
If $ c_P $ is not valid, then the probability of negative decision is at least $ \frac{1}{2y} $ times of the probability of positive decision.

When the computation is finished, if $V$ does not give a decision on $c_P$, $V$ moves input head on the left end-marker and restarts the process of computation and interaction. If $V$ gives a negative decision on $c_P$, the input is rejected. If $V$ gives a positive decision on $c_P$, the input is accepted if $x_l$ is computed as 1, and the input is rejected if $x_l$ is computed as 0.

If the prover is honest, then $Pr[A] = Pr[R]$. If $w \in L$, the probability to accept the input is at least $Pr[A] \cdot \frac{3}{4}$ and the probability to reject the input is at most $Pr[A] \cdot \frac{1}{4} + y \cdot Pr[R]$. Therefore, the total probability to accept the input is at least
\[
\frac{Pr[A] \cdot \frac{3}{4}}{Pr[A] \cdot \frac{3}{4} + Pr[A] \cdot \frac{1}{4} + y \cdot Pr[A]} = \frac{3}{4 \cdot (1+y)},
\]
which can be arbitrarily close to $\frac{3}{4}$ by picking sufficient value $y$. If $w \notin L$ and the prover is honest, then the probability to accept the input is at most $Pr[A] \cdot \frac{1}{4}$ and the probability to reject the input is at least $Pr[A] \cdot \frac{3}{4} + y \cdot Pr[R]$. Therefore, the total probability to reject the input is at least
\[
\frac{Pr[R] \cdot \frac{3}{4} + y \cdot Pr[R]}{Pr[R] \cdot \frac{3}{4} + y \cdot Pr[R] + Pr[R] \cdot \frac{1}{4}} = \frac{3 + 4y}{4 + 4y} > \frac{3}{4}.
\]

If the prover is cheating and provides a finite $c_P$, the probability to reject the input is at least $ \frac{1}{2y} $ times of the probability to accept the input. Therefore, the total probability to reject the input is at least
\[
\frac{\frac{1}{2y}}{1 + \frac{1}{2y}} = \frac{1}{2y + 1},
\]
which can be arbitrarily close to 1 by picking sufficient value $y$. A cheating prover may provide an infinite $c_P$, in which case $V$ does not stop. 

$V$ can execute the algorithm of calculation of value $x_l$ multiple times in each interaction and in case of positive decision for $c_P$ choose the most frequent outcome for $x_l$ as the decision for recognition, thus, increase the probability of correct decision arbitrarily close to 1.
\qed
\end{proof}

Now we switch our focus to the IPSs with two provers. It is known that a 1PFA verifier can simulate the work tape of a TM reliably with high probability by interacting with two provers \cite{FS89}. Since we use this fact later, we present its explicit proof here.

\begin{fact}
	\label{fact:tape-simulation}
	\cite{FS89}
	A 1PFA verifier $V$ can simulate the work tape of a TM reliably with high probability by interacting with provers $P_1$ and $P_2$.
\end{fact}
\begin{proof}
	Let $m$ denote the contents of the infinite to the right work tape including the position of head, where for each $i>0$, $m_i$ denotes $i$-th symbol from the left on work tape. To store the position of head one of $m_i$'s is marked. We accomplish this by doubling the alphabet of work tape, so one symbol can simultaneously store the value and marker of presence of head. 

At the beginning of computation, $V$ secretly picks random values $a$, $b$, and $r_0$, where each of them is between 0 and $q-1$ and $q$ is a predetermined prime number greater than the alphabet of the work tape. The verifier $V$ interacts with provers $P_1$ and $P_2$ and asks them to store $m$ in the following way: for each odd (resp., even) $i$, $V$ asks $P_1$ (resp., $P_2$) to store $m_i$, $r_i$, $s_i$, where $0 \leq r_i \leq q-1$ is picked by $V$ randomly, and $s_i = (m_i \cdot a + r_i \cdot b + r_{i-1}) \mod q$ is a signature. Therefore, each prover stores sequence of triples ($m_i$, $r_i$, $s_i$) for $i$ in ascending order.

To read the contents from the work tape, $V$ scans $m$ from the left to the right, i.e., for each $i$, $ V $ requests from the provers to send the triples ($m_i$, $r_i$, $s_i$). While $V$ scans the input, it checks the correctness of signatures. If $V$ finds a defect, it rejects the input and also identifies the prover giving the incorrect signature as a cheater. 

In order to update the content of the work tape, $V$ picks new random $r_0$, scans the input, for each triple ($m_i$, $r_i$, $s_i$), generates new $r_i$, recalculates $s_i$, and asks the provers to replace the values. If $m_i$ is changed for some $i$, then the values of triple are changed accordingly. Note that one update may include the change of one $m_i$ and the change of the position of work head by one cell, in which case at most 2 sequential $m_i$'s are updated. In this case, $V$ asks the provers to update the contents of the corresponding triples.

The provers cannot learn the values of $a$ and $b$ from the information provided by $V$, since, for each $i>0$, $V$ sends $m_i$ that is derived from the computation, random value $r_i$, and value $s_i$, which has one-to-one correspondence with $r_{i-1}$, which was also chosen randomly. Note that $r_0$ is also randomly chosen and kept in secret.

Let $(m_i, r_i, s_i)$ be the triple provided by $V$ to one of the provers, say $P$, for some $i$, and $(m'_i, r'_i, s'_i)$ be the values provided by $P$, where at least one of the values is different from the one sent by $V$. Assume that $s'_i$ corresponds to $m'_i$ and $r'_i$ correctly. Therefore, $s'_i = (m'_i \cdot a + r'_i \cdot b + r_{i-1}) \mod q$. Then, we have the following relation between $a$ and $b$: $(s'_i - s_i) = ( (m'_i - m_i) \cdot a + (r'_i - r_i) \cdot b + r_{i-1}) \mod q$. Since $q$ is a prime number, exactly $q$ pairs of $(a,b)$'s satisfy this equation, and there are total $q^2$ different pairs of $(a,b)$'s. Thus, since $P$ does not know the values of $a$ and $b$, the probability that $P$ can provide valid values is $\frac{1}{q}$. 

	If both provers are honest, then they send the correct stored values to $V$. This implies that all signature checks will be passed successfully and $V$ accepts the interactions. If a prover changes at least one symbol of the contents that $V$ entrusted to store, the signature check will fail for the triple of changed symbol with probability at least $\frac{q-1}{q}$, in which case $V$ rejects the interaction and identifies the cheating prover. Note that the described protocol works equally well in both IPS models with two provers, because any case of cheating is recognized individually. 
\qed
\end{proof}

Now we can present our next result about the verification of every language.

\begin{theorem}
	\label{thm:two-prover-every-language}
	There is an IPS with two provers ($P_1,P_2,V$) for any language where $ V $ is a 1PFA.
\end{theorem}
\begin{proof}
	A 2PFA verifier, say $V'$, can execute the algorithm for the Corollary \ref{cor:all-k-ary-exp} by interacting with two provers as described in the proof of Fact \ref{fact:tape-simulation} such that the provers reliably store the contents of the counters for coin tosses and processing the number of heads. If the content of the work tape is changed by a prover, $V'$ can catch this with probability at least $\frac{q-1}{q}$, and so rejects the input with the same probability. 

	Let $ w \in \Sigma^* $ be the given input and let $ x_l $ represent its membership bit in $ p_L $. If both provers are honest, then $V'$ correctly performs $64^l$ tosses of $\coinL$ and then processes the number of heads. If the bit $x_l$ (($3l+3$)-th bit in the number of heads) is guessed as 1, the input is accepted, and, if $x_l$ is guessed as 0, then the input is rejected. Therefore, the input is verified correctly with probability $  \frac{3}{4} $. If at least one of the provers is not honest, then any change of the contents stored by the prover is detected with probability at least $\frac{q-1}{q}$, and the input is rejected. This is true if either one or both provers are cheating.

	We can modify $ V' $ and obtain 1PFA verifier $V$. At the beginning of computation, $V$ reads the input from left to right and writes it on the work tape (asks the provers to store it). Then, $V$ implements the rest of the protocol by staying on the right end-marker.
    
    After writing the input on work tape, like in Theorem \ref{thm:weak-verify}, $V$ can execute the algorithm of calculation of value $x_l$ multiple times and increase the probability of correct decision on recognition arbitrarily close to 1.
\qed
\end{proof}

\section{Verification of uncountably many languages}
\label{sec:verification-uncountable}

\subsection{Constant-space verification of nonregular languages}
\label{sec:languages}

In this subsection, we present two nonregular unary languages and one nonregular binary language that can be verified by 2PFAs in quadratic and linear time, respectively. The protocols presented here will be also used in the next subsection.
\begin{theorem}
	\label{thm:usquare}
	$\usqr = \{a^{m^2} \mid  m>0 \}$ is verifiable by a 2PFA in quadratic expected time with bounded error.
\end{theorem}
\begin{proof}
	The protocol is one-way and the verifier expects from the prover a string of the form
\[
	(a^{m}b)^m b
\] 
for the members of the language, where $ m>0 $.

Let $ w=a^n $ be the given input for $ n>3 $ (the decisions on the shorter strings are given deterministically) and let $ y $ be the string provided by the prover. The verifier deterministically checks whether $ y $ is of the form
\[
	y = a^{m_1}ba^{m_2}b \cdots b a^{m_i}b \cdots
\mbox{ or }
	y = a^{m_1}ba^{m_2}b \cdots b a^{m_t}bb
\]
for some $ t>0 $. If the verifier sees a defect on $y$, then the input is rejected.

In the remaining part, we assume that $ y $ is in one of these forms. At the beginning of the computation, the verifier places the input head on the left end-marker and splits computation in four paths with equal probabilities. 

In the first path, the verifier reads $w$ and $y$ in parallel and checks whether $ y $ is finite, i.e.
\[
	y = a^{m_1}ba^{m_2}b \cdots b a^{m_t}b b,
\]
and whether it satisfies the equality $ n=\sum_{j=1}^t m_j $, where $ t>1 $. If one of the checks fails, the input is rejected. Otherwise, it is accepted. 

The second path is very similar to the first path and the following equality is checked:
\[
	n = \sum_{j=2}^t m_j + \sum_{j=1}^t 1,
\]
i.e., the verifier skips $ a^{m_1} $ from $ y $ and counts $ b $'s instead. If the equality is satisfied, the input is accepted. Otherwise, it is rejected.

The computation in the first and second paths is deterministic (a single decision is given in each) and  both paths terminate in linear time.

In the third path, the verifier tries to make the following consecutive comparisons:
\[
	m_1 = m_2, m_3 = m_4, \ldots, m_{2j-1} = m_{2j}, \ldots .
\] 
For each $j$, the verifier can easily determine whether $ m_{2j-1} = m_{2j} < n $ by attempting to move the input head to the right by $m_{2j-1}$ squares and then to the left by $m_{2j}$ squares. If the right end-marker is visited ($ m_{2j-1} \geq n $), or the left end-marker is visited earlier than expected ($ m_{2j}> m_{2j-1} $) or is not visited ($ m_{2j} < m_{2j-1} $), then the comparison is not successful and so the input is rejected. Otherwise, the comparison is successful and the verifier continues with a random walk (described below) before the next comparison, except that if the last comparison is successful, then the input is accepted without making the random walk. 

The aim of the random walk is to determine whether the prover sends a finite string or not, i.e. the prover may cheat by sending the infinite string $  (a^mb)^* $ for some $ m < n $ which passes successfully all comparison tests described above.

The random walk starts by placing the input head on the first symbol of the input and terminates after hitting one of the end-markers. During the random walk, the verifier pauses to read the string $ y $. It is a well known fact that this walk terminates in $ O(n) $ expected number of steps and the probability of ending on the right (resp., the left) end-marker is $ \frac{1}{n} $ (resp., $ 1- \frac{1}{n} $). 

If the walk ends on the left end-marker, then the verifier continues with the next comparison. If the walk ends on the right end-marker, the verifier checks whether the number of $ a $'s in the remaining part of $ y $ is less than $ n $ or not by reading whole input from right to left. If it is less than $ n $, then the input is accepted. Otherwise ($ y $ contains more than $ n $ $ a $'s), the input is rejected. In any case, the computation is terminated with probability $ \frac{1}{n} $ after  the walk.

The fourth path is identical to the third path by shifting the comparing pairs: the verifier tries to make the following consecutive comparisons:
\[
	m_2 = m_3, m_4 = m_5, \ldots, m_{2j} = m_{2j+1}, \ldots .
\] 

Now, we can analyze the overall protocol. If $ n = m^2 $ for some $ m>1 $, then the prover provides $ y = (a^mb)^mb $ and the input is accepted in every path and so the overall accepting probability is 1. Thus, every member is accepted with probability 1. Moreover there will be at most $ m $ random walks and so the overall running time is $ O(n\sqrt{n}) $.

If the input is not a member, then the input is rejected in at least one of the paths. If it is rejected in the first or second path, then the overall rejecting probability is at least $ \frac{1}{4} $. If it is rejected in the third or fourth paths, then the overall rejecting probability cannot be less than $ \frac{3}{16} $ as explained below.

We assume that the input is not rejected in the first and second paths. Then, we know that $ y $ is finite, the number of $a$'s in $ y $ is $n$, and $ y $ is composed by $ m_1 $ blocks. Since $ n $ is not a perfect square, there is at least one pair of consecutive blocks that have different number of $ a $'s. Hence, at least one of the comparisons will not be successful, and, the input will be rejected in one of these paths. Let $l$ be the minimum index such that the comparison of the $ l $-th pair is not successful (in the third or fourth path). Then, $ l < \frac{\sqrt{n}}{2} $. (If not, $ y $ contains at least $ 2 \left\lceil \frac{\sqrt{n}}{2} \right\rceil $ blocks and each of these blocks contains $ m_1 = \left\lceil \sqrt{n} ~ \right\rceil  $ $ a $'s, and, this implies that $ y $ contains more than $ n $ $a$'s.) Then, the maximum accepting probability in the corresponding path is bounded from above by
\[
	 \sum_{i=1}^l \frac{1}{n} \left( 1- \frac{1}{n} \right)^{i-1} = 1- \left( 1- \frac{1}{n} \right)^l < 1 - \left( 1- \frac{1}{n} \right)^{\frac{\sqrt{n}}{2}} \leq \frac{1}{4}.
\]
(Remember that $ n>3 $.)
Therefore, the rejecting probability in the third or fourth path is at least $ \frac{3}{4} $, and so, the overall rejecting probability cannot be less than $ \frac{1}{4} \cdot \frac{3}{4} = \frac{3}{16} $.

The maximum (expected) running time occurs when the prover sends the infinite $ y = (a^mb)^* $ for some $ m>0 $. In this case, the protocol is terminated in the third and fourth paths with probability 1 after $ O(n) $ random walks, and so, the expected running time is quadratic in $ n $, $O(n^2)$.

By repeating the protocol above many times, say $ r>1 $, we obtain a new protocol such that any non-member is rejected with probability arbitrarily close to 1, i.e. $ 1 - \left(1 -  \frac{3}{16} \right)^r $. 
\qed
\end{proof}
\begin{theorem}
	\label{thm:upower64}
	$ \upower  = \{ a^{64^m} \mid m > 0 \} $ is verifiable by a 2PFA with bounded error in quadratic expected time.
\end{theorem}
\begin{proof}
	The proof is very similar to the proof given for $ \usqr $. The protocol is one-way and the verifier expects to receive the string
	\[
		y_n = a b a^{64} b \cdots b a^{64^{k-2}} b a^{64^{k-1}} bb
 	\]
	from the prover for some $ k > 0 $.
	
	Let $ w = a^n $ be the given input for $ n>64 $. (The decisions on the shorter strings are given deterministically.) Let $ y $  be the string sent by the prover. The verifier deterministically checks whether $ y $ is of the form
	\[
		 y = a^{m_1} b a^{m_2} b \cdots b a^{m_i} b \cdots  
	\mbox{ or }
		y = a^{m_1} b a^{m_2} b \cdots b a^{m_t} bb
	\]
	for some $ t>0 $, where $ m_1=1 $. If the verifier sees a defect on $y$, the input is rejected. So, we assume that $ y $ is in one of these forms in the remaining part. 
	
	The verifier splits into three paths with equal probabilities at the beginning of the computation. The first path checks whether $ y $ is finite and
	\[
		n = 1 + 63 \cdot \sum_{j=1}^t m_j.
	\]
	If not, the input is rejected (because $\sum_{j=0}^{k-1} 64^j = \frac{64^k-1}{63}$). Otherwise, the input is accepted.
    
    The second and third paths are very similar to the third and fourth paths in the proof for $ \usqr $. In the second path, the verifier checks 
	\[
		64 \cdot m_{2j-1} = m_{2j}
	\]
	for each $ j>0 $. The random walk part is implemented in the same way. The third path is the same except the comparing pairs: the verifier checks 
	\[
		64 \cdot m_{2j} = m_{2j+1}
	\]
	for each $ j>0 $.
	
	If $ w = a^{64^k} $, then the honest prover sends
	\[
		y_n = a b a^{64} b \cdots b a^{64^{k-1}} bb
	\]
	and the input is accepted in all paths. 
	
	If $ w $ is not a member, then the input is rejected in at least one of the paths. If it is rejected in the first path, then the overall rejecting probability is at least $ \frac{1}{3} $. If the input is accepted in the first path, then the rejecting probability in the second or third path is at least $ \frac{99}{100} $, and so the overall rejecting probability is greater than $ \frac{33}{100} $.
	
	We can use the analysis given in the previous proof. Let $ l $ be the minimum index such that the comparison of $ l $-th pair is not successful (in the second or third path). Then, $ l < \frac{\log_{64} n}{2} = \frac{\log n}{12} $. (If not, $ y $ starts with \[ a b a^{64} b \cdots b a^{64^{2l-3}} b a^{64^{2l-2}} b a^{64^{2l-1}}  \] where $ l > \frac{\log_{64} n}{2} $, and so, $ y $ contains $ \frac{64^{2l}-1}{63} $ $ a $'s, which is greater than $ \frac{n-1}{63} $.) The maximum accepting probability in the corresponding path can be bounded from above by
	\[
		\sum_{i=1}^l \frac{1}{n} \left( 1- \frac{1}{n} \right)^{i-1} = 1- \left( 1- \frac{1}{n} \right)^l < 1 - \left( 1- \frac{1}{n} \right)^{\frac{\log n}{12}} < \frac{1}{100}.
	\]
	(Remember that $ n > 64 $.) Therefore, the rejecting probability in the second or third  path is greater than $ \frac{99}{100} $.
	
	The maximum expected running time is quadratic in $ n $, when the prover sends the infinite $ y = a b a^{64} b \cdots b a^{64^i} b a^{64^{i+1}} b \cdots $. By repeating the protocol many times, we obtain a protocol with better success probability.
\qed
\end{proof}

We continue with the verification of a binary nonregular language, a modified version of $ \tt DIMA $ \cite{DY16A}: 
$
	\dimatwo = \{ 0^{2^0}10^{2^1}10^{2^2}1  \cdots 1 0^{2^{3k-1}}11 (0^{2^{3k}}1)^{2^{3k}} \mid k > 0 \}.
$
\begin{theorem}
	\label{thm:dima2}
	$\dimatwo$ is verifiable by a  sweeping PFA in linear time with bounded error.
\end{theorem}
\begin{proof}
	The protocol is one-way and the verifier expects to receive the same input from the prover. 
	
		Let $ w $ be the given input of the form
	\[
		w = 0^{t_1} 1 0^{t_2} 1 \cdots 1 0^{t_m} 11 0^{t'_1} 1 0^{t'_2} 1 \cdots 1 0^{t'_{m'}} 1,
	\]
	where $ t_1 = 1 $, $ m $ and $m'$ are positive integers, $m$ is divisible by 3, and $ t_i,t'_j>0 $ for $ 1 \leq i \leq m $ and $ 1 \leq j \leq m' $. (Otherwise, the input is rejected deterministically.)
	
	Let $ y $ be the string provided by the prover. When on the left end-marker, the verifier splits into three paths with equal probabilities. In the first path, it checks whether $ w = y $. If not, the input is rejected. Otherwise, the input is accepted. 
	
	In the following part, we assume that $ y = w $. In the second path, the verifier checks whether each 0-block has double length of the previous 0-block before symbols ``11'' and each 0-block has the same length of the previous 0-block after symbols ``11'': When reading $ w $ and $ y $ in parallel, the verifier makes the following comparisons:
	\begin{itemize}
		\item for each $ i \in \{1,\ldots,m-1 \} $, whether $ 2t_i = t_{i+1} $,
		\item whether $ 2 t_m = t'_1 $, and,
		\item for each $ j \in \{1,\ldots,m'-1\} $, whether $ t'_j = t'_{j+1} $.
	\end{itemize}
	If one of the comparisons is not successful, then the input is rejected. If all of them are successful, then we know that the input is of the form
	\[
		w = 0^{2^0}10^{2^1}10^{2^2}1 \cdots 1 0^{2^{3k-1}} 11 (0^{2^{3k}} 1)^{m'}.
	\]
	In the third path, the verifier simply checks whether $ m'=2^{3k} $ or not, i.e. the verifier compares $ t'_1 $ from $ y $ with the number of 1's appearing after ``11'' in $ w $. If not, the input is rejected. Otherwise, it is accepted.
	
	If $ w $ is a member, then the honest prover sends $ y=w $, and the verifier accepts the input in all paths with probability 1.
	
	If $ w $ is not a member, then the verifier rejects the input in one of the paths. If the prover sends $ y \neq w $, then the input is rejected in the first path. If $ y=w $ and then one of the comparisons in the second path may not be successful and so the input is rejected in this path. If all of them are successful, then the comparison in the third path cannot be successful and the input is rejected in this path.
	
	The overall protocol terminates in linear time. By repeating the protocol many times, we obtain a protocol with better success probability.
\qed
\end{proof}

\subsection{Constant-space verification of uncountably many languages}
\label{sec:verifications}
In this subsection, we present two constant-space protocols for verifying uncountably many unary and binary languages.

\begin{theorem}
	\label{thm:uncoutable-sweeping}
    Bounded-error sweeping PFAs can verify uncountably many languages in linear time.
\end{theorem}
\begin{proof}
	Let $ w_k $ be the $ k $-th shortest member of $ \dimatwo $ for $ k>0 $. For any $ I \in \mathcal{I} $, we define the following language:
	\[
	\dimatwo(I) = \{ w_k \mid k>0   \mbox{ and } k \in I \}.
	\]
	We describe a one-way protocol for $\dimatwo(I)$. Let $ w $ be the given input. The verifier determines whether $ w = w_k $ for some $ k>0 $ by using the protocol for $\dimatwo$ with high probability. If not, then the input is rejected. In the remaining part, we continue with 
	\[
		w=w_k=0^{2^0}10^{2^1}10^{2^2}1  \cdots 1 0^{2^{3k-1}}11 (0^{2^{3k}}1)^{2^{3k}}.
	\]	
	
	The verifier attempts to toss $\coinI$ $ 64^k $ times and in parallel processes the total number of heads for determining $ x_k $ in $ p_I $. The verifier asks from the prover to send $ 0^{64^k}1 $ and the prover sends $ y = 0^m1 $ (the input is deterministically rejected if $ y $ is not in this form). 
	
	The verifier splits into two paths with equal probabilities. In the first path, it easily determines whether $ m=64^k $ by passing over the input once (the number of 0's after symbols ``11'' is $ 64^k $). If $m \neq 64^k$, then the input is rejected. Otherwise, the input is accepted.  
	 
	 The second path is responsible for coin-tosses and processing the total number of heads. The verifier performs $ m $ coin-tosses. For counting the heads, the verifier uses the part of $ w_k $ after symbols ``11'' as a read-only counter, which is composed of $ 8^k $ blocks of 0's and length of each block is $ 8^k $. Let $ t $ be the total number of heads:
	 \[
	 	t = i \cdot 8^{k+1} + j \cdot 8^k + q = (8i+j) 8^k+q ,
	 \]
 	 where $ i \geq 0 $, $ j \in \{0,\ldots,7\}$, and $ q < 8^k $. 
	Due to Fact \ref{fact:DY16A}, $ x_k $ is the $ (3k+3) $-th digit of $ bin(t) $ with probability at least $ \frac{3}{4} $. In other words, $ x_k $ is guessed as 1 if $ j \in \{4,\ldots,7\} $, and as 0, otherwise. The verifier sets $ j=0 $ at the beginning. Then, for each head, it reads a symbol 0 from the input and after $ 8^k $ heads it updates $ j $ as $ (j+1) \mod 8 $. If the number of heads exceeds $ 64^k $, then the input is rejected. If not, the decision given is parallel to the value of $ j $: the input is accepted  if $ j \in \{4,\ldots,7\} $ and rejected if $ j \in \{0,\ldots,3\} $.
	
	The verifier operates in sweeping mode and each path terminates in linear time. If $ w $ is a member, then the input is accepted with probability at least $\frac{3}{4}$. If $ w \notin \dimatwo $, then it is rejected with high probability. If $ w \in \dimatwo $ but $ w \notin \dimatwo(I) $, then the input is rejected with probability at least $ \frac{1}{2} \cdot \frac{3}{4} = \frac{3}{8} $. By repeating the protocol, the rejecting probability can get arbitrarily close to 1.
	
	Since the cardinality of set $ \{ I \mid I \in \mathcal{I} \} $ is uncountable, there are uncountably many languages in $ \{ \dimatwo(I) \mid I \in \mathcal{I} \} $, each of which is verified by a bounded-error linear-time sweeping PFA.
\qed
\end{proof}

\begin{theorem}
	\label{thm:uncountable-unary}
	2PFAs can verify uncountably many unary languages with bounded error in quadratic expected time.
\end{theorem}
\begin{proof}
	Here, we use all protocols given in the proofs of Theorems \ref{thm:usquare}, \ref{thm:upower64}, and \ref{thm:uncoutable-sweeping}.

	Let $ w_k $ be the $ k $-th shortest member of $ \upower $ for $ k>0 $. 	For any $ I \in \mathcal{I} $, we define language:
    \[
	\upower(I) = \{ w_k \mid k>0   \mbox{ and } k \in I \}.
	\]
	We construct a verifier for this language.

	Let $ w $ be the given input. By using the protocol given for $ \upower $, the verifier can determine whether $ w \in \upower $ or not. If $ w \notin \upower  $, then the input is rejected with high probability.
	
	In the following part, we assume that $ w = w_k $ for some $ k>0 $. The verifier asks from the prover the following string
	$
		y_k = (a^{8^k}b)^{8^k}b.
	$
	Let $ y $ be the string provided by the prover. The verifier splits into two paths with equal probabilities. In the first path, it checks whether $ y = y_k $ by using the protocol for $ \usqr $. In the second path, the verifier assumes that $ y =y_k $ and implements the part of the protocol for $ \dimatwo(I) $, which is responsible for the coin tosses and for determining whether $ k \in I $ or not by checking the total number of heads. The verifier reads $w$ for $64^k$ coin tosses and $y$ for determining the value $x_k$.
	
	If the prover sends $ y_k $, then the verifier correctly determines whether $ k \in I $ or not with probability at least $ \frac{3}{4} $ in the second path. If $ w \in  \upower(I) $, the honest prover sends $ y_k $ and so the input is accepted with probability 1 in the first path and accepted with probability at least $ \frac{3}{4} $ in the second path.
	
	If $ w \notin \upower(I) $, then the input is rejected with probability at least $ \frac{3}{16} $ if the prover does not send $ y_k $ in the first path, and rejected with probability at least $ \frac{3}{4} $ in the second path if the prover does send $ y_k $. Therefore, the overall rejecting probability is at least $ \frac{3}{32} $. 
	
	Since each called protocol runs no more than quadratic expected time in $ |w| $, the running time is quadratic. By repeating the protocol many times, we obtain a protocol with better success probability.
\qed
\end{proof}

\section*{Acknowledgments} Dimitrijevs is partially supported by University of Latvia projects \linebreak AAP2016/B032 ``Innovative information technologies'' and ZD2018/20546 ``For development of scientific activity of Faculty of Computing''. Yakary{\i}lmaz is partially supported by ERC Advanced Grant MQC.

\bibliographystyle{splncs03}
\bibliography{tcs}

\begin{thebibliography}{10}
\providecommand{\url}[1]{\texttt{#1}}
\providecommand{\urlprefix}{URL }

\bibitem{ADH97}
Adleman, L.M., DeMarrais, J., Huang, M.D.A.: Quantum computability. SIAM
  Journal on Computing  26(5),  1524--1540 (1997)

\bibitem{Bab85}
Babai, L.: Trading group theory for randomness. In: STOC'85: Proceedings of the
  17th Annual ACM Symposium on Theory of Computing. pp. 421--429 (1985)

\bibitem{BGKW88}
Ben-or, M., Goldwasser, S., Killian, J., Wigderson, A.: Multi prover
  interactive proofs: How to remove intractability. In: Proc. of 20th STOC. pp.
  113--131 (1988)

\bibitem{CL89}
Condon, A., Lipton, R.J.: On the complexity of space bounded interactive proofs
  (extended abstract). In: FOCS'89: Proceedings of the 30th Annual Symposium on
  Foundations of Computer Science. pp. 462--467 (1989)

\bibitem{DSY15}
Demirci, H.G., Say, A.C.C., Yakary{\i}lmaz, A.: The complexity of debate
  checking. Theory of Computing Systems  57(1),  36--80 (2015)

\bibitem{DY16A}
Dimitrijevs, M., Yakary{\i}lmaz, A.: Uncountable classical and quantum
  complexity classes. In: Eigth Workshop on Non-Classical Models for Automata
  and Applications. books@ocg.at, vol. 321, pp. 131--146. Austrian Computer
  Society (2016)

\bibitem{DY17}
Dimitrijevs, M., Yakary{\i}lmaz, A.: Uncountable realtime probabilistic
  classes. In: Descriptional Complexity of Formal Systems. LNCS, vol. 10316,
  pp. 102--113. Springer (2017)

\bibitem{DS92}
Dwork, C., Stockmeyer, L.: Finite state verifiers $\mbox{I}$: The power of
  interaction. Journal of the ACM  39(4),  800--828 (1992)

\bibitem{FS89}
Feige, U., Shamir, A.: Multi-oracle interactive protocols with space bounded
  verifiers. In: Structure in Complexity Theory Conference. pp. 158--164 (1989)

\bibitem{Fre83}
Freivalds, R.: Space and reversal complexity of probabilistic one-way
  \mbox{Turing} machines. In: Conference on Fundamentals of Computation Theory.
  pp. 159--170 (1983)

\bibitem{GMR89}
Goldwasser, S., Micali, S., Rackoff, C.: The knowledge complexity of
  interactive proof systems. SIAM Journal on Computing  18(1),  186--208 (1989)

\bibitem{Min61}
Minsky, M.: Recursive unsolvability of \mbox{P}ost's problem of ``tag'' and
  other topics in theory of \mbox{T}uring machines. Annals of Mathematics
  74(3),  437--455 (1961)

\bibitem{SayY17}
Say, A.C.C., Yakary{\i}lmaz, A.: Magic coins are useful for small-space quantum
  machines. Quantum Information {\&} Computation  17(11{\&}12),  1027--1043
  (2017)

\bibitem{FST88}
U.~Feige, A.~Shamir, M.T.: The noisy oracle problem. In: Proceedings of CRYPTO
  88

\end{thebibliography}

\end{document}